\documentclass[12pt]{article}

\usepackage{latexsym,amsmath,amscd,amssymb,graphics,float}
\usepackage{enumerate}

\usepackage{graphicx,lscape}

\usepackage[colorlinks]{hyperref}
\usepackage{url}

\begin{document}

\newenvironment{proof}[1][Proof]{\textbf{#1.} }{\ \rule{0.5em}{0.5em}}

\newtheorem{theorem}{Theorem}[section]
\newtheorem{definition}[theorem]{Definition}
\newtheorem{lemma}[theorem]{Lemma}
\newtheorem{remark}[theorem]{Remark}
\newtheorem{proposition}[theorem]{Proposition}
\newtheorem{corollary}[theorem]{Corollary}
\newtheorem{example}[theorem]{Example}

\numberwithin{equation}{section}
\newcommand{\ep}{\varepsilon}
\newcommand{\R}{{\mathbb  R}}
\newcommand\C{{\mathbb  C}}
\newcommand\Q{{\mathbb Q}}
\newcommand\Z{{\mathbb Z}}
\newcommand{\N}{{\mathbb N}}

\newcommand{\bfi}{\bfseries\itshape}

\newsavebox{\savepar}
\newenvironment{boxit}{\begin{lrbox}{\savepar}
\begin{minipage}[b]{15.5cm}}{\end{minipage}\end{lrbox}
\fbox{\usebox{\savepar}}}

\title{{\bf Complete integrability versus symmetry}}
\author{R\u{a}zvan M. Tudoran}

\date{}
\maketitle \makeatother

West University of Timi\c soara, 
Faculty of Mathematics and Computer Science,\\
Department of Mathematics, 
Blvd. Vasile P\^arvan, No. 4, 
300223--Timi\c soara, Romania.\\
E-mail: {\sf tudoran@math.uvt.ro}.

\begin{abstract}
The purpose of this article is to show that on an open and dense set, complete integrability implies the existence of symmetry. 
\end{abstract}

\medskip

\textbf{AMS 2000}: 37J35; 76M60; 37K10; 70H06.


\section{Introduction}
\label{section:one}

One of the most important results in the literature that study the implication "symmetry $\Rightarrow$ integrability" in the context of a general dynamical system is a classical result due to Lie which says that if one have $n$ linearly independent vector fields $X_1,\dots,X_n$ on $\R^n$ that generates a solvable Lie algebra under commutation: $[X_1,X_j]=c^{1}_{1,j}X_1$, $[X_2,X_j]=c^{1}_{1,j}X_1+c^{2}_{2,j}X_2$, $\dots$, $[X_n,X_j]=c^{1}_{1,j}X_1+c^{2}_{2,j}X_2+\dots+c^{n}_{n,j}X_n$, for $j\in\{1,\dots,n\}$, where $c^{k}_{i,j}$ are the structural constants of the Lie algebra, then the differential equation $\dot{x}=X_1(x)$ is integrable by quadratures (see e.g. \cite{arnold}). Moreover, by proving these quadratures one obtain $n-1$ functionally independent first integrals of the system $\dot{x}=X_1(x)$, and consequently one get his complete integrability (see e.g. \cite{kozlov1}). This result was recently improved by Kozlov, by proving that in fact each of the differential equation $\dot{x}=X_j(x)$ is integrable by quadratures (see \cite{kozlov}). Recall also that the above mentioned solvable Lie algebra integrates to a solvable $n$-dimensional Lie group that acts freely on $\R^n$ and permutes the trajectories of the differential system, and hence is a symmetry group of the set of all trajectories of the system. Note that by "permutation" we mean that the group action maps trajectory to trajectory.

The purpose of this article is to study the converse implication: "integrability $\Rightarrow$ symmetry". The main result of this paper states that for a completely integrable differential system $\dot{x}=X(x)$, one can associate a vector field $\widetilde{X}$, defined on an open and dense set, such that to each $\overline{Y}\in\ker\left( \operatorname{ad}_{\widetilde{X}}\right)$ it corresponds a Lie symmetry of the vector field $X$; more exactly, there exists a real scalar function $\mu=\mu(\overline{Y})$ and a vector field $Y=Y(\overline{Y})$ such that $[X,Y]=\mu \cdot X$. Since the existence of a pair made of a vector field $Y$ and a scalar real function $\mu$ such that $[X,Y]=\mu \cdot X$ is equivalent to the existence of a one-parameter Lie group which permutes the trajectories of the differential system $\dot{x}=X(x)$ (see e.g. \cite{stephani}), one obtains that on an open and dense set, "integrability $\Rightarrow$ symmetry". This result is a generalization of a similar one for planar vector fields (see \cite{llibre}).

\section{Hamilton-Poisson formulation and local normal forms of completely integrable systems}

In this section we recall some results from \cite{tudoran} concerning Hamilton-Poisson formulations of completely integrable systems and their local normal form. These results will be explicitly used in order to prove the main results of this paper. For details on Poisson geometry and (Nambu-)Hamiltonian dynamics, see e.g. \cite{abraham}, \cite{arnoldcarte}, \cite{marsdenratiu}, \cite{holm1}, \cite{holm2}, \cite{holm3}, \cite{ratiurazvan}, \cite{nambu1}, \cite{nambu2}.

Let us consider a $\mathcal{C}^1$ differential system on $\R^n$:
\begin{equation}\label{sys}
\left\{\begin{array}{l}
\dot x_{1}=X_1(x_1,\dots,x_n)\\
\dot x_{2}=X_2(x_1,\dots,x_n)\\
\cdots\\
\dot x_{n}=X_n(x_1,\dots,x_n),\\
\end{array}\right.
\end{equation}
generate by the $\mathcal{C}^1$ vector field $X=X_1 \partial_{x_1}+\dots+X_n \partial_{x_n}$.

Suppose that $C_1,\dots,C_{n-2},C_{n-1}:\Omega\subseteq \R^n\rightarrow \R$ are $n-1$ independent $\mathcal{C}^2$ integrals of motion of \eqref{sys} defined on a nonempty open subset $\Omega\subseteq \R^n$.

Following \cite{tudoran}, the vector field $X$ can be realized on the open set $\Omega\subseteq \R^n$ as the Hamilton-Poisson vector field $X_H \in\mathfrak{X}(\Omega)$ with respect to the Hamiltonian function $H:=C_{n-1}$ and respectively the Poisson bracket of class $\mathcal{C}^{1}$ defined by:
$$
\{f,g\}_{\nu;C_1,\dots,C_{n-2}}dx_1\wedge\dots\wedge dx_n=\nu dC_1\wedge\dots dC_{n-2}\wedge df\wedge dg,
$$
where $\nu\in\mathcal{C}^{1}(\Omega,\R)$ is a given real function (rescaling). 

Recall also from \cite{tudoran} that the above Poisson bracket $\{\cdot,\cdot\}_{\nu;C_1,\dots,C_{n-2}}$ is obtained from the (rescaled) canonical Nambu bracket on $\Omega\subseteq \R^n$ as follows:
$$
\{f,g\}_{\nu;C_1,\dots,C_{n-2}}=\nu \cdot \{C_1,\dots,C_{n-2},f,g\}=\nu \cdot \dfrac{\partial(C_1,\dots,C_{n-2},f,g)}{\partial(x_1,\dots,x_n)}.
$$

For similar Hamilton-Poisson and respectively Nambu-Poisson formulations of completely integrable systems see e.g. \cite{g1}, \cite{g2}, \cite{nambu1}, \cite{nambu2}.

Note that the functions $C_1,\dots,C_{n-2}$ form a complete set of Casimirs for the Poisson bracket $\{\cdot,\cdot\}_{\nu;C_1,\dots,C_{n-2}}$.

Recall that the Hamiltonian vector field $X=X_H\in\mathfrak{X}(\Omega)$ is acting on an arbitrary real function $f\in \mathcal{C}^{k}(\Omega,\R)$, $k\geq 2$ as: 
$$X_{H}(f)=\{f,H\}_{\nu;C_1,\dots,C_{n-2}}\in \mathcal{C}^{1}(\Omega,\R).$$

Consequently, the vector field $X=X_H$ may also be expressed as a Nambu-Hamiltonian vector field on $\Omega$ as follows:
$$
X(f)=\nu\cdot\{C_1,\dots,C_{n-2},f,H\}.
$$

Hence, the differential system \eqref{sys} can be locally written in $\Omega$ as a Hamilton-Poisson dynamical system of the type:
$$\left\{\begin{array}{l}
\dot x_{1}=\{x_1,H\}_{\nu;C_1,\dots,C_{n-2}}\\
\dot x_{2}=\{x_2,H\}_{\nu;C_1,\dots,C_{n-2}}\\
\cdots \\
\dot x_{n}=\{x_n,H\}_{\nu;C_1,\dots,C_{n-2}},\\
\end{array}\right.
$$
or equivalently as a Nambu-Hamiltonian dynamical system:
\begin{equation}\label{sysham}
\left\{\begin{array}{l}
\dot x_{1}=\nu \cdot \dfrac{\partial(C_1,\dots,C_{n-2},x_1,H)}{\partial(x_1,\dots,x_n)}\\
\dot x_{2}=\nu \cdot \dfrac{\partial(C_1,\dots,C_{n-2},x_2,H)}{\partial(x_1,\dots,x_n)}\\
\cdots \\
\dot x_{n}=\nu \cdot \dfrac{\partial(C_1,\dots,C_{n-2},x_n,H)}{\partial(x_1,\dots,x_n)}.\\
\end{array}\right.
\end{equation}

Next result from \cite{tudoran} provide a local normal form of completely integrable systems. This theorem will be the key point for proving that complete integrability implies the existence of symmetry.
\begin{theorem}\label{mainresult}
Let \eqref{sys} be a $\mathcal{C}^{1}$ differential system having a set of $n-1$ independent $\mathcal{C}^2$ conservation laws defined on an open subset $\Omega\subseteq \R^n$. Assume that there exists a $\mathcal{C}^1$ rescaling function $\nu$, nonzero on an open and dense subset $\Omega_0$ of $\Omega$, such that the system \eqref{sys} admits a Hamilton-Poisson realization of the type \eqref{sysham} and the Lebesgue measure of the set $$\mathcal{O}:=\left\{ x=(x_1,\dots,x_n)\in \Omega_{0} \mid \operatorname{div}(X)(x)\cdot\dfrac{\partial(1/\nu,C_{1},\dots,C_{n-2},H)}{\partial(x_1,\dots,x_n)}(x)=0\right\}$$ in $\Omega_{0}$ is zero. 

Then, the change of variables $(x_1,\dots,x_n)\mapsto \Phi(x_1,\dots,x_n)=(u_1,\dots,u_n)$ given by
\begin{equation*}
\left\{\begin{array}{l}
u_1 =1/\nu(x_1,\dots ,x_n)\\
u_2 =C_1(x_1,\dots ,x_n)/\nu(x_1,\dots ,x_n)\\
\cdots \\
u_{n-1} =C_{n-2}(x_1,\dots ,x_n)/\nu(x_1,\dots ,x_n)\\
u_n =H(x_1,\dots ,x_n)/\nu(x_1,\dots ,x_n)\\
ds=-\operatorname{div}(X)dt,
\end{array}\right.
\end{equation*}
in the open and dense subset $\Omega_{00}:=\Omega_{0}\setminus \mathcal{O}$ of $\Omega_0$, transforms the system \eqref{sysham} restricted to $\Omega_{00}$ into the linear differential system $u^\prime_1=u_1, \ u^\prime_2=u_2, \dots, \ u^\prime_n=u_n$, where "prime" stand for the derivative with respect to the new time "$s$". 
\end{theorem}

\begin{remark}
The divergence operator used in the theorem above, is the divergence associated with the standard Lebesgue measure on $\R^n$, namely $$\mathcal{L}_{X}dx_1\wedge\dots\wedge dx_n=\operatorname{div}{X}dx_1\wedge\dots\wedge dx_n,$$ where $\mathcal{L}_{X}$ stand for the Lie derivative along the vector field $X$.
\end{remark}

\begin{remark}\label{okrem}
Recall form \cite{tudoran} that if the rescaling function $\nu=:\nu_{cst.}$ is a constant function, then the Lebesgue measure of the set $\mathcal{O}=\Omega_0=\Omega$ is nonzero in $\Omega_0$, and hence the assumptions of the Theorem \eqref{mainresult} do not hold. In this case, one search for a new $\mathcal{C}^1$ rescaling function $\mu$ defined on an open and dense subset $\Omega_0$ of $\Omega$, such that the vector field $\mu \cdot X$ satisfies the assumptions of the Theorem \eqref{mainresult}. The differential system generated by the vector field $\mu \cdot X$ can also be realized as a Hamilton-Poisson dynamical system with respect to the Poisson bracket  $\{\cdot,\cdot\}_{\mu\cdot\nu_{cst.};C_1,\dots,C_{n-2}}$ and respectively the same Hamiltonian $H$ as for the Hamilton-Poisson realization \eqref{sysham} of the vector field $X$.
\end{remark}

\section{Completely integrable systems admits Lie symmetries}

Let us start this section by recalling a classical result concerning Lie symmetries of dynamical systems. In order to do that, consider a $n$-dimensional dynamical system, $\dot{x}=X(x)$, generated by a $\mathcal{C}^1$ vector field $X=X_1(x_1,\dots,x_n)\partial_{x_1}+\dots+X_n(x_1,\dots,x_n)\partial_{x_n}$, and a one-parameter Lie group of transformations $G$ given by
\begin{equation*}
\left\{\begin{array}{l}
x^{\star}_1 (x_1,\dots,x_n;\varepsilon) =x_1 +\varepsilon \xi_1 (x_1,\dots, x_n)+O({\varepsilon}^2)\\
\cdots \\
x^{\star}_n (x_1,\dots,x_n;\varepsilon) =x_n +\varepsilon \xi_n (x_1,\dots, x_n)+O({\varepsilon}^2),
\end{array}\right.
\end{equation*}
acting on an open subset $U\subseteq \R^n$ with associated infinitesimal generator of class $\mathcal{C}^1$,  $Y=\xi_1(x_1,\dots,x_n)\partial_{x_1}+\dots+\xi_n(x_1,\dots,x_n)\partial_{x_n}\in\mathfrak{X}(U)$.

In the above hypothesis, a classical result (see e.g. \cite{stephani}) states that $G$ is a symmetry group for the dynamical system $\dot{x}=X(x)$ (i.e., his action maps trajectory to trajectory), if and only if $Y$ is a Lie symmetry of $X$, in the sense that $[X,Y]=\mu\cdot X$, for some scalar function $\mu=\mu(x_1,\dots,x_n)$.    

We can state now the main result of this article. This result is a generalization of a similar one for planar vector fields \cite{llibre}.
\begin{theorem}\label{main}
In the hypothesis of Theorem \eqref{mainresult}, to each vector field $\overline{Y}\in\ker\left( \operatorname{ad}_{\widetilde{X}}\right)$, where $\widetilde{X}=u_1 \partial_{u_1}+\dots+u_n \partial_{u_n}$, the vector field $Y={\Phi}^{\star}\overline{Y}$ is a Lie symmetry of the vector field $X$ that generates the dynamical system \eqref{sys}, where ${\Phi}^{\star}\overline{Y}$ is the pull-back of the vector field $Y$ through the diffeomorphism $\Phi$.
\end{theorem}
\begin{proof}
Let us denote by $\overline{X}=\Phi_{\star}X$ the push-forward of the vector field $X$ by the local diffeomorphism $\Phi$ defined be the change of variables $(x_1,\dots,x_n)\mapsto (u_1,\dots,u_n)=\Phi(x_1,\dots,x_n)$. Using the expression of the local diffeomorphism $\Phi$, one obtains that $\overline{X}=h\widetilde{X}$, where the scalar function $h$ is given by $h=-\operatorname{div}(X)\circ\Phi^{-1}$. Recall that in $\Omega_{00}$, $\operatorname{div}(X)\neq 0$, and hence $h\neq 0$.

Let $\overline{Y}$ be an arbitrary vector field such that $\overline{Y}\in\ker\left( \operatorname{ad}_{\widetilde{X}}\right)$. Then we have successively 
\begin{align*}\label{imp}
[\overline{X},\overline{Y}]&=[h\widetilde{X},\overline{Y}]=h[\widetilde{X},\overline{Y}]-\overline{Y}(h)\widetilde{X}=h\operatorname{ad}_{\widetilde{X}}(\overline{Y})-\overline{Y}(h)\widetilde{X}=-\overline{Y}(h)\widetilde{X}\\
&=-\overline{Y}(h)\dfrac{1}{h}\overline{X}=-\dfrac{\overline{Y}(h)}{h}\overline{X}.
\end{align*}

Consequently, by denoting $-\dfrac{\overline{Y}(h)}{h}=\overline{\mu}$, one obtains that 
\begin{equation}\label{ecc}
[\overline{X},\overline{Y}]=\overline{\mu}\overline{X}.
\end{equation}

If one denotes $Y=\Phi^{\star}\overline{Y}$, then using \eqref{ecc}, the following relation hold true
\begin{equation*}\label{ec}
[X,Y]=[\Phi^{\star}\overline{X},\Phi^{\star}\overline{Y}]=\Phi^{\star}([\overline{X},\overline{Y}])=\Phi^{\star}(\overline{\mu}\overline{X})=\Phi^{\star}\overline{\mu}\cdot\Phi^{\star}\overline{X}=\Phi^{\star}\overline{\mu}\cdot X,
\end{equation*}
and hence we get the existence of a scalar function $\mu=\Phi^{\star}\overline{\mu}$ such that $[X,Y]=\mu X$.

\end{proof}

\subsection*{Acknowledgment}
This work was supported by a grant of the Romanian National Authority for Scientific Research, CNCS-UEFISCDI, project number PN-II-RU-TE-2011-3-0103.

\bigskip
\bigskip



\begin{thebibliography}{99}

\bibitem{abraham} {\footnotesize \textsc{R. Abraham and J.E. Marsden}, \textit{Foundations of Mechanics}, Benjamin Cummings, New York 1978. }

\bibitem{arnoldcarte} {\footnotesize \textsc{V.I. Arnold}, \textit{Mathematical methods of classical mechanics}, Graduate Texts in Mathematics, vol. 60, second edition, Springer 1989. }

\bibitem{arnold} {\footnotesize \textsc{V.I. Arnold, V.V. Kozlov and A.I. Neishtadt}, \textit{Mathematical aspects of classical and celestial mechanics}, Itogi Nauki i Tekhniki, Sovr. Probl. Mat. Fundamentalnye Napravleniya, vol. 3, VINITI, Moscow 1985. English transl.: Encyclopedia of Math. Sciences, vol. 3, Springer-Verlag, Berlin 1989. }

\bibitem{g1} {\footnotesize \textsc{A. Ay, M. G\"urses and K. Zheltukhin}, Hamiltonian equations in $\R^3$, \textit{J. Math. Phys.}, 44(2003), 5688--5705. }

\bibitem{llibre} {\footnotesize \textsc{J.Gin\'e and J. Llibre}, On the planar integrable differential systems, \textit{Z. Angew. Math. Phys.}, 62(4)(2011), 567--574. }

\bibitem{g2} {\footnotesize \textsc{M. G\"urses, G.S. Guseinov and K. Zheltukhin}, Dynamical systems and Poisson structures, \textit{J. Math. Phys.}, 50(2009), 112703. }

\bibitem{holm1} {\footnotesize \textsc{D.D. Holm}, \textit{Geometric Mechanics I: Dynamics and Symmetry}, World Scientific: Imperial College Press, Singapore 2008. }

\bibitem{holm2} {\footnotesize \textsc{D.D. Holm}, \textit{Geometric Mechanics II: Rotating, Translating and Rolling}, World Scientific: Imperial College Press, Singapore 2008. }

\bibitem{holm3} {\footnotesize \textsc{D.D. Holm, T. Schmah and C. Stoica}, \textit{Geometric Mechanics and Symmetry: From Finite to Infinite Dimensions}, Oxford Texts in Applied and Engineering Mathematics, Oxford University Press, 2009. }

\bibitem{kozlov} {\footnotesize \textsc{V.V. Kozlov}, Remarks on a Lie theorem on the integrability of differential equations in closed form, \textit{Differential Equations},41(4)(2005), 588--590. Translated from Differentsial'nye Uravneniya, 41(4)(2005), 553--555.}

\bibitem{kozlov1} {\footnotesize \textsc{V.V. Kozlov}, \textit{Simmetriya, topologiya i rezonansy v gamil'tonovoi mekhanike (Symmetry, topology, and resonances in Hamiltonian mechanics)}, Izhevsk, 1995.}

\bibitem{marsdenratiu} {\footnotesize \textsc{J.E. Marsden and T.S. Ratiu}, \textit{Introduction to mechanics and symmetry}, Texts in Applied Mathematics, vol. 17, second edition, second printing, Springer, Berlin 1999. }

\bibitem{nambu1} {\footnotesize \textsc{Y. Nambu}, Generalized Hamiltonian mechanics, \textit{Phys. Rev. D},7(8)(1973), 2405--2412. }

\bibitem{ratiurazvan} {\footnotesize \textsc{T.S. Ratiu, R.M. Tudoran, L. Sbano, E. Sousa Dias and G. Terra}, \textit{Geometric Mechanics and Symmetry: the Peyresq Lectures; Chapter II: A Crash Course in Geometric Mechanics}, pp. 23--156, London Mathematical Society Lecture Notes Series, vol. 306, Cambridge University Press 2005. }

\bibitem{stephani} {\footnotesize \textsc{H. Stephani}, \textit{Differential equations: their solutions using symmetries}, Cambridge University Press, Cambridge 1989. }

\bibitem{nambu2} {\footnotesize \textsc{L. Takhtajan}, On foundation of the generalized Nambu mechanics, \textit{Comm. Math. Phys.},160(1994), 295--315. }

\bibitem{tudoran} {\footnotesize \textsc{R.M. Tudoran}, A normal form of completely integrable systems, \textit{J. Geom. Phys.}, 62(5)(2012), 1167--1174. }

\end{thebibliography}
\end{document}